\documentclass[conference,letterpaper]{IEEEtran}
\IEEEoverridecommandlockouts
\usepackage{xcolor}
\usepackage[utf8]{inputenc} 
\usepackage[T1]{fontenc}
\usepackage{url}
\usepackage{ifthen}
\usepackage{cite}
\usepackage{amssymb}
\usepackage{amsfonts}
\usepackage{amsthm}
\usepackage{enumitem}
\usepackage[cmex10]{amsmath} 
\usepackage{mathrsfs}

\newtheorem{theorem}{Theorem}
\newtheorem{lemma}{Lemma}
\newtheorem{definition}{Definition}
\newtheorem{corollary}{Corollary}
\newtheorem{proposition}{Proposition}

\theoremstyle{definition}
\newtheorem{remark}{Remark}
 
\theoremstyle{definition}

\newcommand{\mc}{\mathcal}

\newcommand{\allanswers}{A_{[0:N-1]}}
\newcommand{\allqueries}{Q_{[0:N-1]}}
\newcommand{\allfiles}{W_{[0:K-1]}}
\newcommand{\allqueriesinstantiation}{q_{[0:N-1]}}
\newcommand{\expec}{{\mathbb E}}
\newcommand{\length}{\mathscr{L}}

\def \extended {1}

\begin{document}
\title{On the Optimal Message Size in PIR Under Arbitrary Collusion Patterns} 


\author{%
  \IEEEauthorblockN{Guru S. Dornadula*, Manikya Pant*, Gowtham R. Kurri, Prasad Krishnan}
  \thanks{* Equal Contribution}
  \IEEEauthorblockA{International Institute of Information Technology, Hyderabad \\
                    Hyderabad, India.\\ 
                    Email: \{guru.dornadula@research., manikya.pant@research., gowtham.kurri@, prasad.krishnan@\}iiit.ac.in}
}

\maketitle

\begin{abstract}
	A private information retrieval protocol (PIR) scheme under an arbitrary collusion pattern $\mc{P}$ enables a client to retrieve one message from a library of $K$ equal-sized messages duplicated in $N$ servers, while keeping the index of the desired message private from any colluding set in $\mc{P}$. The efficiency of a PIR protocol is measured by its rate, defined as the ratio of the message size to the total download cost, and the supremum of all the achievable rates is the capacity of PIR. Although achieving high rates typically requires sufficiently large message sizes, smaller message sizes also desirable due to reduced implementation complexity and fewer constraints. By characterizing the capacity-achieving schemes, Tian, Sun, and Chen (2019) showed that the optimal message size for uniformly decomposable PIR schemes under no-collusion setting is $N-1$. However, comparable results are not yet available for more general collusion settings.

    In this work, we present a complete characterization of the properties of capacity-achieving decomposable PIR schemes under arbitrary collusion patterns. Building on this characterization, we derive a general lower bound on the optimal message size for capacity-achieving uniformly decomposable PIR schemes under an arbitrary collusion pattern $\mc{P}$, expressed in terms of the hitting number of a newly defined family of subsets of servers determined by the collusion pattern $\mc{P}$. Finally, we specialize the lower bound to several important classes of collusion patterns, including $T$-collusion, disjoint collections of colluding sets, cyclically $T$-contiguous collusion, and disjoint collections of cyclically contiguous colluding sets. For the last two collusion patterns, we present matching achievable schemes that attain the corresponding bounds, thereby providing a complete characterization of the optimal message size.    
\end{abstract}

\section{Introduction}
A private information retrieval (PIR) system under arbitrary collusion pattern $\mathcal{P} = \{\mathcal{T}_1, \mathcal{T}_2, \ldots, \mathcal{T}_M\}$ consists of $N$ servers hosting a library of equal-sized $K$ messages. A client in this system seeks to download one message from the library without revealing the identity of the desired message to any colluding set $\mc{T}_m\in\mc{P}$. A query-response protocol executed in this system between the client and the servers is known as a PIR scheme. The central metric to determine the performance of a PIR scheme is its \textit{rate}, which is defined as the ratio of the size of the message (in bits) to the expected total quantum of download from the servers (in bits). The supremum of all the achievable rates is the capacity of PIR. The capacity of PIR under information-theoretic privacy has been characterized in several settings through matching achievability and converse results, including the no-collusion setting~\cite{SunJ17}, the $T$-collusion setting where any $T$ servers may collude~\cite{SunJ18}, and the general setting of arbitrary collusion patterns $\mathcal{P} = \{\mathcal{T}_1, \mathcal{T}_2, \ldots, \mathcal{T}_M\}$~\cite{YaoLK21}.

Typically, in any PIR protocol, we write each message as an $L$-length tuple over some alphabet ${\cal X}$, where $L$ is known as the \textit{sub-packetization} parameter and ${\cal X}$ is the \textit{message alphabet}. The \textit{message-size} of a PIR protocol is then the product $L\log|{\cal X}|$. For instance, the capacity-achieving scheme presented in \cite{SunJ17} requires the message-size to be $N^K\log|{\cal X}|$ bits, with $|{\cal X}|=2$ being sufficient. For large values of $N$ and $K$, this message-size becomes impractically large. PIR schemes with small message-sizes are preferable, as a small message-size implies lower complexity of implementation, as well as the increased versatility. Characterizing the minimum message-size for PIR and designing schemes with such optimal message-size has been an ongoing effort in research on PIR. In \cite{sjARBITRARYmessagelengthNK-1}, it was shown that, for the non-collusion setting, the message-size required for any capacity-achieving scheme is $N^{K-1}$, under the condition that the maximum download for any realization of the server-queries cannot exceed the ratio of the message-size to capacity. A capacity-achieving scheme with message-size $N^{K-1}$ was also shown in \cite{sjARBITRARYmessagelengthNK-1}. 

The capacity-achieving scheme for $T$-collusion PIR ($T$-PIR with $2\leq T\leq N$) in \cite{SunJ18} required the sub-packetization $L=N^K$, with $|{\cal X}|=\Omega(N^2T^{M-2})$. Subsequently, it was shown in \cite{zhang_TIT2019_optimalSubpack_LinearPIRCollusion} that a sub-packetization $L=\mathsf{gcd}(N,T)\left(\frac{N}{\mathsf{gcd}(N,T)}\right)^{K-1}$ is necessary, for any capacity-achieving $T$-collusion PIR scheme satisfying some special linearity conditions and under the constraint that the total size of the responses from the servers remains the same for all queries. Following this, the works \cite{XuZhang_2018_ISIT_OptimalSubpackO(N)fieldsize_Tcoll,Xuetal_2019_CapAchievingTPIRschemeusingMDSarraycodes,Xuetal_2022_TPIR_Binaryfield_SubfieldSubcodes} have presented capacity-achieving PIR schemes with this small sub-packetization.

In \cite{TianSC19}, Tian, Sun and Chen presented a novel capacity-achieving scheme (the \textit{TSC scheme}) for the no-collusion setting. Remarkably, the TSC scheme uses a message-size of only $N-1$. This reduction is enabled by two key relaxations: (a) the scheme allows different total download sizes across different query realizations, and (b) the PIR rate is defined using the \textit{expected} download size rather than the maximum download size, as in \cite{sjARBITRARYmessagelengthNK-1}. Under these relaxations, it was shown in \cite{TianSC19} that $N-1$ is the optimal message-size among all capacity-achieving schemes within the class \textit{uniformly} decomposable schemes. The class of decomposable schemes is rich, subsuming linear schemes, and indeed almost all known PIR schems are linear. The lower bound proof in \cite{TianSC19} relies on a complete characterization of capacity-achieving decomposable schemes for PIR under no-collusion, obtained through a careful analysis of the conditions under which the inequalities appearing in the converse proof of \cite{SunJ17} hold with equality. 

We extend the aforementioned the line of work by focusing on the optimal message size in PIR under arbitrary collusion patterns for replicated servers. In particular, our main contributions are as follows.
\begin{itemize}[leftmargin=*]
\item We completely characterize the capacity-achieving decomposable schemes for PIR under arbitrary collusion patterns by identifying the necessary and sufficient conditions on their algebraic structure (Theorem~\ref{thm:characterizing-properties}). This characterization is obtained through a refined analysis of the converse for the PIR capacity under arbitrary collusion~\cite{YaoLK21}, leveraging and extending similar ideas previously developed for the no-collusion setting in \cite{TianSC19}. Under appropriate specializations, these conditions recover the characterizations obtained in \cite[Theorem~2]{MikiMM24} and \cite[Theorem~5]{miki2025necessary}.
\item By exploiting the algebraic structure induced by those conditions in a non-trivial manner, we then derive a general lower bound on the optimal message size for capacity-achieving uniformly decomposable PIR schemes under an arbitrary collusion pattern $\mc{P}$. This lower bound is expressed in terms of the hitting number of a newly defined family of subsets of servers determined by the collusion pattern $\mc{P}$ (Theorem~\ref{thm:minimum-message size}).
\item Finally, we specialize the lower bound to several classes of collusion patterns, including $T$-collusion~\cite{SunJ18}, disjoint collections of colluding sets~\cite{JiaSJ17}, cyclically $T$-contiguous collusion~\cite{SunJ18}, and disjoint collections of cyclically contiguous colluding sets. For the last two settings, we present matching achievable schemes that attain the corresponding bounds, thereby providing a complete characterization of the optimal message size under these collusion patterns.
\end{itemize}

\section{System Model and Preliminaries}
\textit{Notation.} For integers $s\leq t$ we denote the set  $\{s,s+1,\hdots, t\}$ as $[s:t]$.   
Let $\cal S$ denote a set indexing items $X_i:i\in \cal S$. For a subset $\cal T\subseteq \cal S$,  we denote the subset of items $X_i:i\in \cal T$ as $X_{\cal T}$. 

{The scheme considered in~\cite{TianSC19} and the ones used in this paper for special cases of T-contiguous and disjoint T-contiguous are not balanced. The minimum message size for PIR under T-collusion was shown in~\cite{KruglikKDW26} was shown to be N-T but for a balanced scheme }

As in \cite{SunJ17,TianSC19}, we consider a system consisting of $N$ replicated servers indexed by $[0:N-1]$, each containing $K$ files or messages denoted by $W_i:i\in[0:K-1]$. We assume $W_i:i\in[0:K-1]$ are independent random variables and uniformly distributed over ${\cal X}^L$, for some alphabet ${\cal X}$, where $L$ is a positive integer. Thus, we have $H(\allfiles)=\sum_{i=0}^{K-1}H(W_i)=KL\log|{\cal X}|$. We refer to the term $L\log|{\cal X}|$ as the \textit{message size}, as in \cite{TianSC19}. We consider PIR under an arbitrary collusion pattern $\mathcal{P} = \{\mathcal{T}_1, \mathcal{T}_2, \ldots, \mathcal{T}_M\}$, where each $\mathcal{T}_m\subseteq [0:N-1]$, $m\in[1:M]$, denotes a maximal set of servers that may collude. We make the following assumptions on the collusion pattern $\mathcal{P}$ without loss of generality:
\begin{enumerate}[leftmargin=*]
    \item If $\mathcal{T}\subseteq \mathcal{P}$, then every subset of $\mathcal{T}$ is also a colluding set.
    \item Every server index appears in at least one element of $\mathcal{P}$.
    \item No two distinct server indices $i,j\in[0:N-1]$ always appear together in all elements of $\mathcal{P}$; otherwise such indices can be merged and treated as a single server.
\end{enumerate}

The client wishes to retrieve the message $W_k$ from the servers while keeping the index $k$ private from any colluding set of servers specified by $\mathcal{P}$.

Formally, a PIR scheme under the collusion pattern $\mathcal{P}$ consists of the following steps:

\begin{itemize}[leftmargin=*]
	\item \textit{Query Generation:} The client possesses a private random key (unknown to servers), denoted by $F$ which takes values in a set ${\cal F}$ according to some distribution. The client also possesses collections ${\cal Q}_n:n\in[0:N-1]$, where ${\cal Q}_n$ refers to the set of all possible queries to server $n$. For $n\in[0:N-1]$, the client generates query $Q_n^{[k]}:n\in[0:N-1]$ as $Q_n^{[k]}\triangleq \phi_n(k,F)$, where $\phi_n$ is a deterministic function which takes the desired file index $k$ and the private key $F$ as its input, given as $~\phi_n:[0:K-1]\times {\cal F}\to {\cal Q}_n.$ Query $Q_n^{[k]}$ is then sent to server $n$, for $n\in[0:N-1]$.

	\item \textit{Responses from Servers:} For each $n\in[0:N-1]$, the server generates the response-length $\length_n$ as a deterministic function (taking non-negative values) of the received query $Q_n^{[k]}\in{\cal Q}_n$, i.e., $\length_n\triangleq\ell_n(Q_n^{[k]})$. The server $n$ then transmits a response $A_n^{[k]}\triangleq \varphi_n(Q_n^{[k]},\allfiles)$ of length $\length_n$ over the coded-symbol alphabet ${\cal Y}$, where $\varphi$ is some deterministic function given as 
	      $
	      ~\varphi_n: {\cal Q}_n\times {\cal X}^{KL}\to {\cal Y}^{\length_n}.
	      $
	      
	\item Then, the client retrieves the desired file $W_k$ satisfying \emph{correctness} and \emph{$\mathcal{P}$-privacy} using some reconstruction function applied to the received responses  $A_n^{[k]}:n\in[0:N-1]$, $F$, and $k$. Formally, correctness requires $H(W_k|\allanswers^{[k]},\allqueries^{[k]})=0$, and the privacy condition is
	      \begin{align}
	      P_{Q_{\mathcal{T}_{m}}^{[k]}}(q_{\mathcal{T}_{m}})
	      = P_{Q_{\mathcal{T}_{m}}^{[k']}}(q_{\mathcal{T}_{m}}),\;
	      \forall\, k,k' \in [0:K-1], 
        \end{align}
 for each $\mathcal{T}_{m} \in \mathcal{P}$.       
\end{itemize}
The rate of a information retrieval scheme is defined as the ratio $\frac{L\log|{\cal X}|}{\log{|\cal Y}|\cdot\expec(\sum_{n=0}^{N-1}\length_n)}.$ A rate $R$ is said to be achievable if there exists retrieval protocol that satisfies the decoding and privacy constraint.

The capacity of PIR under collusion pattern $\mathcal{P}$ was derived in \cite{YaoLK21} as 
\begin{align}\label{eqn:capacity-arbitrary}
	C_{\mc{P}}
	=
	\left(
	1 + \frac{1}{S^{\ast}}
	+ \left(\frac{1}{S^{\ast}}\right)^{2}
	+ \cdots
	+ \left(\frac{1}{S^{\ast}}\right)^{K-1}
	\right)^{-1},
\end{align}
where \(S^{\ast}\) is the optimal value of the following linear
programming problem:
\begin{align}
	\max_{\mathbf{y}} \quad & \mathbf{1}_{N}^{T} \mathbf{y} \label{eq:LP1}       \\\   
	\text{subject to} \quad & \mathbf{B}^{T}_{\cal{P}} y \le \mathbf{1}_{M}, \nonumber \\
	                        & \mathbf{y} \ge \mathbf{0}_{N}, \nonumber           
\end{align}
where \(\mathbf{B}_{\cal{P}}\) is the \(N \times M\) incidence matrix representing the collusion pattern $\mathcal{P} = \{\mathcal{T}_1, \mathcal{T}_2, \ldots, \mathcal{T}_M\}$. Specifically, if server \(n\in[0:N-1]\) belongs to the 
$\mathcal{T}_m$, $m\in[1:M]$, then the \((n,m)\)-th entry of 
\(\mathbf{B}_{P}\) is \(1\); otherwise, it is \(0\).

We recall key definitions from \cite{TianSC19} for a general class of PIR scheme (adapted for PIR with collusion pattern $\mathcal{P}$) as we need them for our lower bounds on the minimum message size.

\begin{definition}[\!\!{\cite[Definitions~2 and 3]{TianSC19}}]
	\label{defn:decomposabilityanduniform}
	A PIR with collusion pattern $\mathcal{P}$ scheme is called \textit{decomposable}, if ${\cal Y}$ is a finite abelian group (with operation $\oplus$), and for each $n\in[0:N-1]$ and $q\in{\cal Q}_n$, the response function $\varphi_n(q, W_{[0:K -1]})$ can be written as
	{\small \begin{align*}&\varphi_n(q,W_{[0:K -1]})\\&=\left(
		\varphi^{(q)}_{n,0}(W_{[0:K -1]}),\varphi^{(q)}_{n,1}(W_{[0:K -1]}),\hdots,\varphi^{(q)}_{n,\ell_n(q)-1}(W_{[0:K -1]})\right),
		\end{align*}}
	where
	$\varphi^{(q)}_{n,i}(\allfiles)
	= \varphi^{(q)}_{n,i,0}(W_0)\oplus\varphi^{(q)}_{n,i,1}(W_1)\oplus\cdots\oplus \varphi^{(q)}_{n,i,K-1}(W_{K-1}),\forall i\in[0:\ell_n(q)-1]$ 
	where each
	$\varphi^{(q)}_{n,i,k}$ is a mapping ${\cal X}^L \to {\cal Y}$. Further, a \textit{uniformly decomposable} PIR scheme is a decomposable scheme where $\varphi^{(q)}_{n,i,k}(W_k), \forall n,i,k$, either have uniform distribution over ${\cal Y}$ or are deterministic.
\end{definition}

The work \cite{TianSC19} also provided a simple way of denoting the action of the response function of each server $n$ as a matrix-vector pseudo-product, which we now recall. These notations will be used throughout this work. We denote $\allfiles\cdot G_n^{(q)}\triangleq \varphi_n(q,\allfiles),$  where $\allfiles$ is viewed as a $K$-length vector over alphabet ${\cal X}^L$, and $G_n^{(q)}$ is viewed as a matrix of size $K\times \ell_n(q)$ whose elements are $G^{(q)}_{n,i,k}$ represent the component functions, written as: 
$W_k\cdot G^{(q)}_{n,i,k}\triangleq \varphi^{(q)}_{n,i,k}(W_k)$. Also,  $G_{n|\mathcal{A}}^{(q)}$ is used to denote a submatrix of $G_n^{(q)}$, with the rows corresponding to the subset $\mathcal{A}$ of the messages removed.

We need the following fractional version of Shearer's lemma~\cite{MadimanT10}, together with its necessary and sufficient conditions for equality~\cite{JakharKCP25}, to characterize capacity-achieving schemes under arbitrary collusion patterns.
	
	
	\begin{proposition}[{\!\!
				\cite[Proposition~II]{MadimanT10}
				\cite[Corollary~2]{JakharKCP25}}]\label{shearers lemma}
					
					Let $ \alpha: \mathcal{P}\to \mathbb{Q}_{+}$ be any fractional covering with respect to a family $\mathcal{P}$ of subsets of $[1:n] , i.e,\sum \limits_{P \in \mathcal{P} : i\in \mathcal{P}} \alpha(P) \geq 1 , \forall i$. For jointly distributed random variables $X_{1}, \dots , X_{n}$, 
					\begin{align}\label{eqn:fractional covering}
						H(X_{[1:N]}) \leq \sum \limits_{P \in \mathcal{P}} \alpha(P) H(X_{P}), 
					\end{align}
					 with equality if and only if $X_{i}$'s  for  $i$ such that $\sum\limits_{P \in \mathcal{P}: i\in \mathcal{P}} \alpha(P)= 1$ are mutually independent, and $X_{i}$  for $i$  such that $\sum\limits_{P \in \mathcal{P}: i\in \mathcal{P}} \alpha(P) > 1$  are constants.
	\end{proposition}

\section{The Characterizing Properties of Capacity-Achieving PIR Schemes Under Arbitrary Collusion}

In this section, we characterize the necessary and sufficient conditions for capacity-achieving decomposable PIR schemes under arbitrary collusion pattern $\mc{P}$, by refining the converse proof given in \cite{YaoLK21} and focusing on the tightness of the underlying inequalities. This approach generalizes that of \cite[Section~IV-A]{TianSC19}, which obtained necessary conditions for capacity-achieving decomposable PIR schemes in the no-collusion setting.

Consider a collusion pattern $\mathcal{P} = \{\mathcal{T}_1, \mathcal{T}_2, \ldots, \mathcal{T}_M\}$. Let $\mathbf{x}^{*}$ be an optimal solution to the following linear program:
		\begin{align}
			\min_{\mathbf{x}} \quad & \mathbf{1}_{M}^{T}\mathbf{x}    \label{eqn:LPinx}                           \\
			\text{subject to} \quad 
			                        & \mathbf{B}_{\mathcal{P}}\mathbf{x} \ge \mathbf{1}_{N},\nonumber \\
			                        & \mathbf{x} \ge \mathbf{0}_{M}\nonumber.                             
		\end{align}
The optimal value of this linear program is equal to $S^*$, i.e., the optimal value of the linear program in \eqref{eq:LP1}~\cite{YaoLK21}. If $x^*_m=0$ for $m\in\mathcal{M}_\phi$, then the capacity of this collusion pattern $\mc{P}$ is same as that of the reduced pattern obtained by removing $\mathcal{T}_m$ for all $m\in\mathcal{M}_\phi$, i.e., $\mathcal{P}'=\{\mc{T}_m: \mc{T}_m\in\mc{P}, m\notin\mathcal{M}_\phi\}$. Therefore, without loss of generality, we assume throughout this paper that for a given collusion pattern $\mathcal{P} = \{\mathcal{T}_1, \mathcal{T}_2, \ldots, \mathcal{T}_M\}$, there exists an optimal solution $x^*$ for the linear program in \eqref{eqn:LPinx} such that $x^*_m>0$ for all $m\in[1:M]$.

\begin{theorem}\label{thm:characterizing-properties}
	Any decomposable PIR scheme under collusion pattern $\mathcal{P} = \{\mathcal{T}_1, \mathcal{T}_2, \ldots, \mathcal{T}_M\}$ achieves capacity $C_\mc{P}$ in \eqref{eqn:capacity-arbitrary} if and only if it satisfies the properties {\bf{P1}}-{\bf{P4}} below.  
	    
	Consider a query set $\allqueriesinstantiation=(q_0,q_1,\dots,q_{N-1})$ such that $P(\allqueries^{[k]}=\allqueriesinstantiation)>0$, with the corresponding answers $A_0^{(q_0)},A_{1}^{(q_1)},\dots,A_{N-1}^{(q_{N-1})}$ and the answer coding matrices $G_0^{(q_0)},G_1^{q_1},\dots,G_{N-1}^{(q_{N-1})}$.\\

	\begin{enumerate}[label=\bf{P\arabic*.},start=1]
		
		\item \underline{Independence of the server responses:} The $N$ random variables ${A}_{0}^{(q_{0})}$, ${A}_{1}^{(q_{1})}$, .... ,${A}_{N-1}^{(q_{N-1})}$
		      are independent.
		\item \underline{Uniform distribution of the server responses: } \\
         $A_n^{(q_n)}$ is uniformly distributed over $\mathcal{Y}^{\ell_n(q_n)}$,  where  $\ell_{n}(q_n)$ is the length of the answer from the $n^{th}$ server, i.e., $H\!\left(A_n^{(q_n)}\right)=\ell_n(q_n)\,\log |\mathcal{Y}|$.
		\item \underline{Identical information for the residuals:}
		      For each $\mc{T}_m\in\mc{P}$, the $N$ random variables 
              \begin{align}
              R_n\triangleq W_{0:k-1,\,k+1:K-1} \cdot G^{(q_n)}_{n|k}, n\in[0:N-1],
              \end{align}
		      are deterministic functions of $R_{\mc{T}_m}=(R_i: i\in\mc{T}_m)$,

		\item \underline{ Independence of the requested message signals:} \\
		      The N random variables 
              \begin{align}
              W_k \cdot G^{(q_n)}_{n|0:k-1,\,k+1:K-1}, n\in[0:N-1],
              \end{align}
		      are independent.
		      
		      \
	\end{enumerate}
\end{theorem}

\begin{remark}
    The properties \textbf{P1}, \textbf{P3}, and \textbf{P4} in Theorem~\ref{thm:characterizing-properties} recover the necessary conditions for capacity-achieving decomposable PIR schemes in the no-collusion setting, i.e., $\mc{P}=\{\{0\}, \{1\}, \dots, \{N-1\}\}$~\cite[Lemmas~1 and 2]{TianSC19}.
\end{remark}

\begin{remark}\label{remark2}
In \cite[Theorem~5]{miki2025necessary} and \cite[Theorem~2]{MikiMM24}, necessary and sufficient conditions for capacity-achieving linear PIR schemes under $T$-collusion (including the no-collusion case $T=1$) are obtained as follows: 
\begin{align}
                            \text{rank} \ Q_{[0:N-1]}^{(k)}[\mc{I}] &= \sum_{n=0}^{N-1} \text{rank} \ Q_n^{(k)}[\mc{I}],\label{eqn:tcolrec1}\\
                            \text{rank} \ Q_{[0:N-1]}^{(k)}[\bar{k}] &=\max_{\substack{\tau\subseteq [0:N-1]:\\ |\tau|=T}}\text{rank} \ Q_\tau^{(k)}[\bar{k}] \label{eqn:tcolrec2},           
\end{align}
where $k\in\mc{I}\subseteq [0:K-1]$, $ \bar{k}=[0:K-1]\setminus\{k\}$, $Q_n^{(k)}=(Q_n[0], Q_n[1], \dots, Q_n[K-1])$ is the query to $n^{\text{th}}$-server, and
\begin{align}
Q_n^{(k)}[C]&=(Q_n^{(k)}[i]: i\in C), C\subseteq [0:K-1],\\
Q_{D}^{(k)}[S]&=(Q_i[S]: i\in D)^T, D\subseteq [0:N-1]. 
\end{align}
When specialized to linear PIR schemes under $T$-collusion, Theorem~\ref{thm:characterizing-properties} recovers \eqref{eqn:tcolrec1} and \eqref{eqn:tcolrec2}, see Appendix~\ref{recovering_rank_equations} for details.
\end{remark}

\begin{proof}[Proof Sketch of Theorem~\ref{thm:characterizing-properties}]
     The properties $\textbf{P1}$ and \textbf{P2} directly follow from \cite[Lemma~1]{TianSC19}, which was for the no-collusion setting, since their proofs do not rely on the privacy constraint and use only the correctness and structural properties of the PIR scheme. The proof for the properties \textbf{P3} and \textbf{P4} rely on the privacy constraint, among other things. In particular, we derive {\bf{P3}} and {\bf{P4}} as the necessary and sufficient conditions for the following inequality to hold for any decomposable PIR scheme under arbitrary collusion pattern $\mc{P}$, for any $k\in[1:K-1]$: 
    \begin{align}
		  & S^*I(W_{[k:K-1]};\allanswers^{[k-1]}|W_{[0:k-1]},F)\nonumber                   \\
		  & ~~\geq  I(W_{[k+1:K-1]};\allanswers^{[k]}|W_{[0:k]},F)+L\log{|\mathcal{X}|}, 
	\end{align}
    where $S^*$ is the optimal value of the linear programming problem in \eqref{eq:LP1}. Although property \textbf{P4} coincides with its counterpart in the no-collusion setting~\cite[Lemma~2]{TianSC19}, it arises here via fundamentally different argument.  Specifically, our proof leverages the necessary and sufficient conditions for equality in the fractional version of Shearer's lemma (see Proposition~\ref{shearers lemma}), established recently in \cite{JakharKCP25}. A detailed proof is given in 
    \if \extended 1%
    Appendix \ref{proof:thm1}.
\fi
\if \extended 0%
    the extended version~\cite[Appendix~A]{DornadulaPKP26}.
\fi
\end{proof}

\section{A Lower Bound on the Message Size and Optimality for Cyclically Contiguous Collusion}
In this section, we derive a general lower bound on the optimal message size for capacity-achieving PIR schemes under arbitrary collusion patterns. We then specialize this bound to several classes of collusion patterns, including $T$-collusion~\cite{SunJ18}, disjoint collections of colluding sets~\cite{JiaSJ17}, cyclically $T$-contiguous collusion~\cite{SunJ18}, and disjoint collections of cyclically contiguous colluding sets. For the last two collusion patterns, we present matching achievable schemes that attain the corresponding bounds, thereby exactly characterizing the optimal message size. We also present several examples to illustrate the general bound. Let $L^*\log{|\mc{X}|}$ denote the optimal message size  of capacity-achieving PIR schemes under collusion pattern $\mc{P}$. 

\subsection{A General Lower Bound on the Optimal Message Size}
Recall from the discussion preceding Theorem~\ref{thm:characterizing-properties} that, without loss of generality, for a given collusion pattern $\mathcal{P} = \{\mathcal{T}_1, \mathcal{T}_2, \ldots, \mathcal{T}_M\}$, there exists an optimal solution $\mathbf{x}^*$ for the linear program in \eqref{eqn:LPinx} such that $x^*_m>0$ for all $m\in[1:M]$.
\begin{theorem}\label{thm:minimum-message size}
	The optimal message size of any uniformly decomposable capacity-achieving PIR scheme under collusion pattern $\mathcal{P} = \{\mathcal{T}_1, \mathcal{T}_2, \ldots, \mathcal{T}_M\}$ is lower bounded as
  \begin{align}
      L^*\log{|\mc{X}|}&\geq \alpha\left(\mc{F}(\mc{P})\right)\log{|\mc{Y}|}\label{eqn:messbound}\\
      &\geq \alpha\left(\mc{F}(\mc{P})\right)\label{eqn:messbound1},
  \end{align}
where $\alpha(\mc{F}(\mc{P}))$ denotes the hitting number of the family $\mc{F}(\mc{P})$, i.e., the minimum cardinality of a set that intersects every set in the family
\begin{align}
\mathcal{F}(\mc{P})= \Big\{ S \subseteq [0\!:\!N\!-\!1]& :\;
\exists\, \mc{T}_m \in \mathcal{P},\ \mc{T}_m \subseteq S,\ 
\forall i \in \mc{T}_m,\nonumber\\
&\exists\, \mc{T}_r \in \mathcal{P}\ \text{s.t. } 
\mc{T}_r \subseteq S \setminus \{i\} \Big\} \label{Family_for_hitting}.
\end{align}
\end{theorem}
\begin{proof}[Proof Sketch]
    Let the desired-message-dependent part of the answer $A_n^{(q_n)}$ be denoted by $M_n=W_k \cdot G^{(q_n)}_{n|0:k-1,\,k+1:K-1}$, $n\in[0:N-1]$. We first show that the family $\mc{F}(\mc{P})$ in \eqref{Family_for_hitting} characterizes the collection of all sets $S\subseteq [0:N-1]$ for which it is not possible that all entropies $H(M_i)$, $i\in S$ are equal to zero. The main novelty in our proof lies in identifying a \emph{common} structure for such sets $S$ that applies to \emph{arbitrary} collusion patterns $\mc{P}$. We achieve this by exploiting, in a non-trivial manner, the algebraic structure induced by the properties \textbf{P1}-\textbf{P4} for capacity-achieving schemes to derive a unified characterization. We then show that the optimal message size is lower bounded by the minimum number of non-zero entropies among $H(M_n)$, $n\in[0:N-1]$, which in turn corresponds to the hitting number of $\mc{F}(\mc{P})$. A detailed proof is given in 
    \if \extended 1%
    Appendix \ref{appendix:Thm2proof}.
\fi
\if \extended 0%
    the extended version~\cite[Appendix~B]{DornadulaPKP26}.
\fi
\end{proof}
We next specialize the above bound in \eqref{eqn:messbound} to two classes of collusion patterns, namely, $T$-collusion and collusion pattern consisting of disjoint collections of colluding sets, which lead to the following corollaries.  A detailed proof of Corollary~\ref{corollory_1} is given in 
    \if \extended 1%
    Appendix \ref{appndx:cor1proof}.
\fi
\if \extended 0%
    the extended version~\cite[Appendix~C]{DornadulaPKP26}.
\fi

					\begin{corollary} \label{corollory_1}
						For any uniformly decomposable capacity-achieving PIR scheme under $T$-collusion, i.e., $\mc{P}=\{S\subseteq [0:N-1]: |S|=T\}$, we have
						\begin{align}\label{eqn:lowbnd-T-PIR}
							L^*\log|\mathcal{X}| \ge (N-T)\log|\mathcal{Y}|. 
						\end{align}  
					\end{corollary}
\begin{remark}
    For $T=1$, the bound in \eqref{eqn:lowbnd-T-PIR} reduces to $(N-1)\log{|\mc{Y}|}$, recovering the optimal message size for the no-collusion setting~\cite[Theorem~3]{TianSC19}.
\end{remark}	
\begin{remark}
    We compare the lower bound on optimal message-size in $T$-PIR in \eqref{eqn:lowbnd-T-PIR} with the minimum message-size in PIR with MDS-coded servers~\cite{Zhouetal_2019_TIT_MDS_optimalmsgsize_lowerupcost,Zhuetal_2020_MDS_optimalmessagesize_higher_upcost}. Although the capacity of PIR under $T$-collusion with $N$ replicated servers coincides with that of PIR with $(N,T)$-MDS-coded servers~\cite{BanawanU18}, this does not necessarily imply that the corresponding message size requirements must also coincide. In particular, the optimal message size in capacity-achieving PIR schemes for $(N,T)$-coded servers is $\textbf{lcm}(N-T,T)$~\cite{Zhouetal_2019_TIT_MDS_optimalmsgsize_lowerupcost,Zhuetal_2020_MDS_optimalmessagesize_higher_upcost}, whereas Corollary~\ref{corollory_1} establishes a lower bound of $N-T$ for PIR under $T$-collusion. These two quantities coincide if and only if $T$ divides $N$ and $N\geq 2T$. It remains an open question whether the lower bound $N-T$ is actually achievable for PIR under $T$-collusion.  
\end{remark}

Consider $N$ servers partitioned into $M$ disjoint sets. For each $m\in[1:M]$, the $m^{\text{th}}$-set consists of $N_m$ servers, among which any $T_m\leq N_m$ servers may collude. This pattern is referred to as $(N_1,\dots,N_M ; T_1,\dots,T_M)$-collusion~\cite{JiaSJ17}.
				
					\begin{corollary}\label{corollary2}						    
						For any uniformly decomposable capacity-achieving PIR scheme under $(N_1,\dots,N_M ; T_1,\dots,T_M)$-collusion, we have
						\begin{align}\label{eqn:lowbnd-disjoint-T-PIR}
							L^*\log|\mathcal{X}| \ge \left(\sum_{m=1}^M(N_m-T_m)+(M-1)\right)\log|\mathcal{Y}|. 
						\end{align}
					\end{corollary}
					
                    A detailed proof of Corollary~\ref{corollary2} is given in 
                    \if \extended 1%
    Appendix \ref{appdnx:cor2proof}.
\fi
\if \extended 0%
    the extended version~\cite[Appendix~D]{DornadulaPKP26}.
\fi
                    
                    \subsection{The Optimal Message Size Under Cyclically Contiguous Collusion}
                For collusion patterns involving cyclically contiguous servers, we completely characterize the optimal message size.
					\begin{theorem}\label{thm:T-contiguous}
						For any uniformly decomposable capacity-achieving PIR scheme under cyclically $T$-contiguous collusion, i.e., $\mc{P}=\{\{0,1,\dots,T-1\}, \{1,\dots,T\}, \dots, \{N-1,0,\dots,T-2\}\}$, we have
						\begin{align}
							L^*\log{|\mc{X}|} \geq \left(\left\lceil{\frac{N}{T}}\right\rceil-1\right)\log{|\mc{Y}|}, 
						\end{align}
                        where $\lceil x\rceil$ is the smallest integer larger than $x$. 

                        Moreover, when $T$ divides N, we have
                        \begin{align}
                            L^*\log{|\mc{X}|} = \left(\left\lceil{\frac{N}{T}}\right\rceil-1\right)\log{|\mc{Y}|}.
                        \end{align}
					\end{theorem}
A detailed proof of Theorem~\ref{thm:T-contiguous} is given in 
\if \extended 1%
    Appendix~\ref{appndx:proofofthm3}.
\fi
\if \extended 0%
    the extended version~\cite[Appendix~E]{DornadulaPKP26}.
\fi
\begin{remark}\label{remark5}
    In \cite[Appendix~D]{SunJ18_MDS} and \cite{YaoLK21}, the proposed capacity-achieving PIR schemes under cyclically $T$-contiguous collusion requires a message size of $N^K$, see Appendix~\ref{apndx:messagesize-contiguous} for details. In contrast, Theorem~\ref{thm:T-contiguous} implies that our achievable scheme uses the optimal message size of $\lceil \frac{N}{T} \rceil-1$, which is exponentially smaller than $N^K$.
\end{remark}

\begin{remark}
    The capacity of PIR under $T$-collusion is equal to that under cyclically $T$-contiguous collusion~\cite{YaoLK21}. However,  while Theorem~\ref{thm:T-contiguous} shows that the optimal message size for capacity-achieving PIR schemes under cyclically $T$-contiguous collusion is $\left\lceil{\frac{N}{T}}\right\rceil-1$, Corollary~\ref{corollory_1} implies that the optimal message size for PIR under $T$-collusion is $L^*\log|\mathcal{X}|\geq (N-T)$, which is strictly larger than $\left\lceil{\frac{N}{T}}\right\rceil-1$.
\end{remark}
                    
				Motivated by $(N_1,\dots,N_M ; T_1,\dots,T_M)$-collusion~\cite{JiaSJ17}, we define an analogous collusion pattern for cyclically contiguous collusion over disjoint sets of servers. Consider $N$ servers partitioned into $M$ disjoint sets. For each $m\in[1:M]$, the $m^{\text{th}}$-set consists of $N_m$ servers, among which any cyclically $T_m$-contiguous servers may collude. We refer to this collusion pattern as $(N_1,\dots,N_M ; T_1,\dots,T_M)$-cyclically contiguous collusion.
					\begin{theorem}\label{them:disjointcontiguous}
						For any uniformly decomposable capacity-achieving PIR scheme under $(N_1,\dots,N_M ; T_1,\dots,T_M)$-cyclically contiguous collusion, we have 
						\begin{align}
							L^*\log{|\mc{X}|} \geq \left(\sum_{i=1}^{M}\left\lceil{\frac{N_i}{T_i}}\right\rceil-1\right)\log{|\mc{Y}|}. \label{thm_disjoint_t_conti} 
						\end{align} 
                        Moreover, when $T_m$ divides $N_m$, for $m\in[1:M]$, we have
                        \begin{align}
							L^*\log{|\mc{X}|} = \left(\sum_{i=1}^{M}\left\lceil{\frac{N_i}{T_i}}\right\rceil-1\right)\log{|\mc{Y}|}.
						\end{align} 
					\end{theorem}
                    A detailed proof of Theorem~\ref{them:disjointcontiguous} is given in 
                \if \extended 1%
    Appendix~\ref{appndx:thm4proof}.
\fi
\if \extended 0%
    the extended version~\cite[Appendix~F]{DornadulaPKP26}.
\fi
                    \begin{remark}\label{remark7}
                    Since \cite{YaoLK21} provides an achievable scheme under arbitrary collusion pattern, one can infer an achievable scheme for $(N_1,\dots,N_M ; T_1,\dots,T_M)$-cyclically contiguous collusion with message size $L
=
\left(
\sum_{i=1}^{M}
\frac{ \left( \prod_{j=1}^{M} T_j \right) N_i }{ T_i }
\right)^{K}$. In contrast, Theorem~\ref{them:disjointcontiguous} implies that our achievable scheme uses the optimal message size of $\left(\sum_{i=1}^{M}\left\lceil{\frac{N_i}{T_i}}\right\rceil-1\right)$, which is exponentially smaller than that. See Appendix~\ref{apndx:messagesize-contiguous} for details.
                    \end{remark}

\subsection{An Illustrative Example}

$N=5$ and $\mc{P} = \{\{0,1,2\},\{0,3\},\{1,3\},\{2,3\},\{4\}\} $.

 The corresponding incidence matrix is
 \[
\mathbf{B}_{\mathcal{P}} =
\begin{bmatrix}
1 & 1 & 0 & 0 & 0 \\
1 & 0 & 1 & 0 & 0 \\
1 & 0 & 0 & 1 & 0 \\
0 & 1 & 1 & 1 & 0 \\
0 & 0 & 0 & 0 & 1
\end{bmatrix}.
\]

 The optimal solution to the linear program in \eqref{eqn:LPinx} is $x^{*}=[ \frac{2}{3},\frac{1}{3},\frac{1}{3},\frac{1}{3},1]^{T}$.
 So, $x^*_m>0$ for all $m\in[1:5]$.
 Let 
 \begin{align}
 \mc{F}(\mc{P})=\mc{F}'(\mc{P})\cup \{A:\exists F\in\mc{F}'(\mc{P}) \ s.t. \ F\subseteq A\},
 \end{align}
 where $\mc{F}'(\mc{P})$ is the collection of all minimal sets of $\mc{F}(\mc{P})$. A set $F\in\mc{F}(\mc{P})$ is said to be minimal if no proper subset of $\mc{F}$ belongs to $\mc{F}(\mc{P})$. By definition, $\mc{F}(\mc{P})$ is upward closed, and hence $\alpha(\mc{F}(\mc{P}))=\alpha(\mc{F}'(\mc{P}))$. Scanning over all the subsets of $[0:N-1]$, it is easy to verify that
\begin{align}
\mc{F}'(\mc{P})\!=\! \{\{0,3,4\}, \{1,3,4\}, \{2,3,4\}, \{0,1,2,3\}. \{0,1,2,4\}\}
\end{align}

The minimum cardinality of a set that intersects every set in $\mc{F}'(\mc{P})$ is $2$, e.g., the set \{1,4\} intersects all $F\in\mc{F}'(\mc{P})$. Therefore, $\mc{F}(\mc{P})=\mc{F}'(\mc{P})=2$.

\if \extended 1%

					\appendices

                    \section{Proof of Theorem~\ref{thm:characterizing-properties}}\label{proof:thm1}

We first present two lemmas and show that the equality conditions in those two lemmas completely characterize the necessary and sufficient conditions for capacity-achieving PIR schemes under an collusion pattern $\mathcal{P}$.
\begin{lemma}\label{lemma:Ctian}
	For any PIR scheme under collusion pattern $\mathcal{P}$ and any $k\in[0:K-1]$, we have
	\begin{align}\label{eqn:tianlemma}
		\hspace{-0.5cm}I(W_{[0:K-1]\setminus\{k\}};\allanswers^{[k]}|W_k,F)\leq \left(\frac{L}{R}-L\right)\log{|\mathcal{X}|}. 
	\end{align} 
	Moreover, the properties {\bf{P1}} and {\bf{P2}} are necessary and sufficient for any decomposable scheme to satisfy \eqref{eqn:tianlemma} with equality for all $k\in[0:K-1]$.
\end{lemma}
We note that the proof of Lemma~\ref{lemma:Ctian} follows exactly as in \cite[Proof of Lemma~1]{TianSC19}. This is because the proof relies solely on the following properties of a PIR scheme: (i) $W_0,W_1,\dots,W_{K-1}$ are mutually independent, (ii) the retrieved message is correct, (iii) the answer functions $\allanswers$ are deterministic functions of $(\allfiles,F)$. In particular, the privacy condition is not required for the proof. We remark that property {\bf{P2}} is not explicitly stated in \cite[Lemma~1]{TianSC19}, as the equality conditions for one of the inequalities appearing in the proof, specifically, $H(A_n^{(q_n)})\leq \log_2{\mathcal{Y}^{\ell_n}}$, $n\in[0:N-1]$ in \cite[Proof of Lemma~1]{TianSC19}, were not analyzed. Nevertheless, it is straightforward to see that {\bf{P2}} is a necessary and sufficient condition for this inequality to hold with equality. 

For a permutation function $\pi:[0:K-1]\rightarrow [0:K-1]$, let $\pi(S)=\{\pi(i):i\in S\}$, for $S\subseteq [0:K-1]$.

\begin{lemma}\label{lemma:iterative-arbitrary-pir}
	Let $\pi:[0:K-1]\rightarrow [0:K-1]$ be a permutation function. For any PIR scheme under collusion pattern $\mathcal{P} = \{\mathcal{T}_1, \mathcal{T}_2, \ldots, \mathcal{T}_M\}$, and any $k\in\{1,2,\dots,K-1\}$, we have
	\begin{align}\label{eqn:mainlemma}
		  & S^{*}I(W_{\pi([k:K-1])};\allanswers^{\pi([k-1])}|W_{\pi([0:k-1])},F)\nonumber                   \\
		  & ~~\geq  I(W_{\pi([k+1:K-1])};\allanswers^{[k]}|W_{\pi([0:k])},F)+L\log{|\mathcal{X}|}, 
	\end{align}
	where $S^*$ is the optimal value of the linear program
		\begin{align}
			\min_{\mathbf{x}} \quad & \mathbf{1}_{M}^{T}\mathbf{x}                              \nonumber \\
			\text{subject to} \quad 
			                        & \mathbf{B}_{\mathcal{P}}\mathbf{x} \ge \mathbf{1}_{N},\nonumber \\
			                        & \mathbf{x} \ge \mathbf{0}_{M} \label{LP2} .                            	\end{align}
	
	Moreover, the properties {\bf{P3}} and {\bf{P4}} are necessary and sufficient for any decomposable PIR scheme to satisfy \eqref{eqn:mainlemma} with equality for all $k\in[0:K-1]$.
\end{lemma}
\begin{proof}[Proof of Lemma~\ref{lemma:iterative-arbitrary-pir}]

                    Let $\mathbf{x^*}=(x_1,\dots,x_M)$ be an optimal solution to the linear program in \eqref{LP2}. Since \eqref{LP2} has integer coefficients, it admits an optimal solution with rational entries, and hence $\mathbf{x^*}$ can be chosen to be rational~\cite[Corollary~1.3.1]{ScheinermanU97}. Then $S^*=\sum_{m=1}^Mx_i$. To establish ~\eqref{eqn:mainlemma}, without loss of generality, let us consider the identity permutation function, i.e., $\pi(k)=k$, for all $k\in[0:K-1]$. We begin by considering the following information inequalities.
					
					\begin{align}
						  & S^{*} I(W_{[k:K-1]};A_{0:N-1}^{[k-1]}|W_{[0:k-1]},F) \nonumber                                      \\
						  & \quad = \sum_{m=1}^M x_{m} I(W_{[k:K-1]};A_{0:N-1}^{[k-1]}|W_{[0:k-1]},F)                 \\
						  & \quad \geq  \sum_{m} x_{m}I(W_{[k:K-1]};A_{\mathcal{T}_{m}}^{[k-1]}|W_{[0:k-1]},F) \label{eq:3prime} \\
                        & \quad = \sum_{m=1}^M x_{m}H(A_{\mathcal{T}_{m}}^{[k-1]}|W_{[0:k-1]},F) \label{eq:2}\\
                        & \quad = \sum_{m=1}^M x_{m}H(A_{\mathcal{T}_{m}}^{[k-1]}|W_{[0:k-1]},Q^{[k-1]}_{\mathcal{T}_m})\label{eq:3}\\
                         & \quad = \sum_{m=1}^M x_{m}H(A_{\mathcal{T}_{m}}^{[k]}|W_{[0:k-1]},Q^{[k]}_{\mathcal{T}_m})\label{eq:4}\\
                          & \quad = \sum_{m=1}^M x_{m}H(A_{\mathcal{T}_{m}}^{[k]}|W_{[0:k-1]},F)\label{eq:5}\\
						  & \quad \geq H(A_{0:N-1}^{[k]}|W_{[0:k-1]},F) \label{eq:equality_for_shearer}                   \\ 
						  & \quad = I(A_{0:N-1}^{[k]};W_{k:K-1}| W_{[0:k-1]},F) \label{eq:6}                              \\
						  & \quad = I(A_{0:N-1}^{[k]}, W_{k}; W_{k:K-1}| W_{[0:k-1]},F) \label{eq:7}                      \\
						  & \quad = L \log{(|\mathcal{X}|)} + I(W_{k+1:K-1}, A_{0:N-1}^{[k]} | W_{[0:k]}, F),             
					\end{align}

					where ~\eqref{eq:2} and \eqref{eq:6} follow  because the answers are deterministic functions of the messages $\allfiles$ and the random key $F$, \eqref{eq:3} and \eqref{eq:5} follow from the Markov chain $(A_n^{[k]},\allfiles)-Q_n^{[k]}-F$, where $k$ is the index of the desired message, \eqref{eq:4} is a consequence of the privacy requirement, which ensures that $(Q_{\mathcal{T}_{m}}^{[k-1]},A_{\mathcal{T}_{m}}^{[k-1]},W_{[0:k-1]})$ and $(Q_{\mathcal{T}_{m}}^{[k]},A_{\mathcal{T}_{m}}^{[k]},W_{[0:k-1]})$ have identical probability distributions, \eqref{eq:equality_for_shearer} follows by an application of the fractional version of shearer lemma, i.e., Proposition~\ref{shearers lemma}, since $(x_1,\dots,x_M)$ forms a fractional covering with respect to $\mathcal{P}$ in view of the inequality constraints in the linear program in \eqref{LP2}, and \eqref{eq:7} follows from the correctness requirement.
					
					For ~\eqref{eqn:mainlemma} to hold with equality for any decomposable PIR scheme, the inequalities in \eqref{eq:2} and~\eqref{eq:equality_for_shearer} must hold with equality.  Specifically, for \eqref{eq:3prime} to be an equality, we must have, with $k=1$ and $\mathcal{T}_m \in \mathcal{P}$, and using the notation $G_{\mathcal{T}|0}^{(q_{\mathcal{T}})}=(G_{i|0}^{(q_i)}:i\in\mathcal{T})$ and $G_{[0:N-1]|0}^{(\allqueriesinstantiation)}=(G_{i|0}^{(q_i)}:i\in[0:N-1])$, 
					
					\begin{align}
						0 & =I(W_{[1:K-1]};\allanswers^{[0]} \mid W_{0},F)   \\                                                                   
						&\hspace{12pt}-I(W_{[1:K-1]};A_{\mathcal{T}_{m}}^{[0]} \mid W_{0},F)\nonumber\\
						  & =H(\allanswers^{[0]} \mid W_{0},F)                                                                                  
						-H(A_{\mathcal{T}_{m}}^{[0]} \mid W_{0},F)
						\label{iterative-arbitrary_pir-proof11}\\
						  & =H(\allanswers^{[0]} \mid A_{\mathcal{T}_{m}}^{[0]},W_{0},F)                                                        
						\label{iterative-arbitrary_pir-proof32}\\
						  & =H(\allanswers^{[0]} \mid A_{\mathcal{T}_{m}}^{[0]},W_{0},F,\allqueries^{[0]})                                      
						\label{iterative-arbitrary_pir-proof12}\\
						  & =H(\allanswers^{[0]} \mid A_{\mathcal{T}_{m}}^{[0]},W_{0},\allqueries^{[0]})\label{iterative-arbitrary_pir-proof13} \\
						  & = \sum_{q_{0:N-1}}\Pr(Q_{0:N-1}= q_{0:N-1}) \nonumber                                                               \\ 
						  & \quad \cdot H(A_{0:N-1}^{[0]} \mid A_{\mathcal{T}_{m}}^{[0]}, W_{0},                                                
						Q_{0:N-1}=q_{0:N-1}) \label{eq:step1} \\
						  &=\sum_{\allqueriesinstantiation}P_{\allqueries^{[0]}}(\allqueriesinstantiation)\cdot\nonumber\\
      &\hspace{23pt}H(W_{[1:K-1]}\cdot G_{[0:N-1]|0}^{(\allqueriesinstantiation)}~|~W_{[1:K-1]}\cdot G_{\mathcal{T}_m|0}^{(q_{\mathcal{T}_m})},W_0,\nonumber\\
      &\hspace{130pt}\allqueries^{[0]}=\allqueriesinstantiation)\label{lemma:iterative-t-pir-proof14}\\
       &=\sum_{\allqueriesinstantiation}P_{\allqueries^{[0]}}(\allqueriesinstantiation)\cdot\nonumber\\
      &\hspace{32pt}H(W_{[1:K-1]}\cdot G_{[0:N-1]|0}^{(\allqueriesinstantiation)}~|~W_{[1:K-1]}\cdot G_{\mathcal{T}_m|0}^{(q_{\mathcal{T}_m})})\label{lemma:iterative-t-pir-proof15},
					\end{align}

					where \eqref{iterative-arbitrary_pir-proof32} follows because the answers are deterministic functions of the messages $\allfiles$ and the random key $F$, \eqref{iterative-arbitrary_pir-proof12} follows because the queries $\allqueries^{[0]}$ are deterministic functions of $F$ given that $0$ is the index of the desired message, \eqref{iterative-arbitrary_pir-proof13} follows from the Markov chain $(\allanswers^{[0]},\allfiles)-\allqueries^{[0]}-F$ , \eqref{lemma:iterative-t-pir-proof14}  utilizes the fact that the component functions $W_0\cdot G^{(q_n)}_{n|[1:K-1]}$ can be subtracted from the answers given $W_0$, \eqref{lemma:iterative-t-pir-proof15} holds because $W_0$ becomes independent of all other random variables after the corresponding component functions are eliminated from the answers, and the dependence on $\allqueriesinstantiation$ is fully absorbed into the answer function matrices $G_{n}^{(q_n)}$, $n\in[0:N-1]$. This implies that, for $\allqueriesinstantiation$ such that $P_{\allqueries}(\allqueriesinstantiation)>0$, we have
					$ H(W_{[1:K-1]}\cdot G_{[0:N-1]|0}^{(\allqueriesinstantiation)}|W_{[1:K-1]}\cdot G_{\mathcal{T}_{m}|0}^{(q_{\mathcal{T}_{m}})})=0$,
					for all $\mathcal{T}_{m}\in \mathcal{P}$. This, in turn, implies that, for any $\mathcal{T}_m\in \mathcal{P}$, $W_{[1:K-1]}\cdot G_{\mathcal{T}|0}^{(q_{\mathcal{T}})}$ can determine $W_{[1:K-1]}\cdot G_{n}^{(q_n)}$, for all $n\in[0:N-1]$. Since the entire argument above can be repeated by interchanging the message index $0$ with an arbitrary $k\in[0:K-1]$, using the appropriate permutation of $[0:K-1]$, we obtain the property {\bf{P3}}.
					
					For \eqref{eq:equality_for_shearer} to hold with equality, the equality conditions in the fractional version of Shearer's lemma, i.e., Proposition~\ref{shearers lemma}, imply that the random variables $A_0^{[k]},A_1^{[k]},\dots,A_{N-1}^{[k]}$ are conditionally independent given $(W_{[0:k-1]},F)$, for $k\in[1:K-1]$. Equivalently, for $k=K-1$, 
					\begin{align}
						0 & =\left(\sum_{n=0}^{N-1}H(A_n^{[K-1]} \mid W_{[0:K-2]},F)\right)\nonumber                                               \\
						  & \hspace{2.3cm}-H(\allanswers^{[K-1]} \mid W_{[0:K-2]},F)\label{lemma:iterative-t-pir-proof35}                          \\
						  & =\sum_{n=0}^{N-1}H(A_n^{[K-1]} \mid W_{[0:K-2]},\allqueries^{[K-1]})\nonumber                                          \\
						  & \hspace{12pt}-H(\allanswers^{[K-1]} \mid W_{[0:K-2]},\allqueries^{[K-1]})\label{lemma:iterative-arbitrary-pir-proof31} \\
						  & =\sum_{\allqueriesinstantiation}                                                                                       
						P_{\allqueries}(\allqueriesinstantiation)\nonumber\\
						  & \hspace{12pt}\Bigg(                                                                                                    
						\sum_{n=0}^{N-1}
						H(A_n^{[K-1]} 
						\mid W_{[0:K-2]},\allqueries^{[K-1]}=\allqueriesinstantiation)\nonumber\\
						  & \hspace{12pt}                                                                                                          
						-H(\allanswers^{[K-1]} 
						\mid W_{[0:K-2]},\allqueries^{[K-1]}=\allqueriesinstantiation)
						\Bigg)\\
						  & =\sum_{\allqueriesinstantiation}                                                                                       
						P_{\allqueries^{[K-1]}}(\allqueriesinstantiation)
						\Big(\sum_{n=0}^{N-1}
						H(W_{K-1}\cdot G_{n\mid[0:K-2]}^{(q_n)})\nonumber\\
						  & \hspace{35pt}                                                                                                          
						-H(W_{K-1} \cdot 
						G_{[0:N-1]\mid[0:K-2]}^{(\allqueriesinstantiation)})
						\Big)\label{lemma:iterative-t-pir-proof33},
					\end{align}

					where \eqref{lemma:iterative-arbitrary-pir-proof31} follows for similar reasons as the steps from 
					\eqref{iterative-arbitrary_pir-proof12} through \eqref{iterative-arbitrary_pir-proof13} follow, and 
					\eqref{lemma:iterative-t-pir-proof33} follows analogously to the steps from \eqref{eq:step1} through \eqref{lemma:iterative-t-pir-proof15}. 
					This implies that the random variables 
					$W_{K-1}\cdot G_{n\mid[0:K-2]}^{(q_n)}$, for $n\in[0:N-1]$, are independent.
					
					Since the above argument beginning at \eqref{lemma:iterative-t-pir-proof35} can be repeated by interchanging the message index $K-1$ with any arbitrary $k\in[0:K-1]$, using the appropriate permutation of $[0:K-1]$, we obtain the property \textbf{P4}. This completes the proof of Lemma~\ref{lemma:iterative-arbitrary-pir}. 
      \end{proof}
					
					Using the inequalities \eqref{eqn:tianlemma} and \eqref{eqn:mainlemma} in Lemmas~\ref{lemma:Ctian} and \ref{lemma:iterative-arbitrary-pir}, respectively, the converse proof of PIR under arbitrary collusion can be completed as follows. Starting from Lemma~\ref{lemma:Ctian} with $k=1$ and applying \eqref{eqn:mainlemma} from Lemma~\ref{lemma:iterative-arbitrary-pir} iteratively for $k=2$ to $k=K-1$, we have:
					\begin{align}
  & L\left(\frac{1}{R}-1\right)\log{|\mathcal{X}|}
    \geq I(W_{[1:K-1]};\allanswers^{[0]}|W_0,F)\nonumber \\
  & \geq \frac{1}{S^{*}}I(W_{[2:K-1]};\allanswers^{[1]}|W_{[0:1]},F)
    +\frac{1}{S^{*}}L\log{|\mathcal{X}|} \\
  & \geq \frac{1}{S^{*}}\Bigg(\frac{1}{S^{*}}I(W_{[3:K-1]};
    \allanswers^{[2]}|W_{[0:2]},F)\nonumber \\
  & \hspace{12pt}+\frac{1}{S^{*}}L\log{|\mathcal{X}|}\Bigg)
    +\frac{1}{S^{*}}L\log{|\mathcal{X}|} \\
  & =\left(\frac{1}{S^{*}}\right)^2
    I(W_{[3:K-1]};\allanswers^{[2]}|W_{[0:2]},F)\nonumber \\
  & \hspace{12pt}
    +\left(\frac{1}{S^{*}}+\left(\frac{1}{S^{*}}\right)^2\right)
    L\log{|\mathcal{X}|} \\
  & \geq \cdots \nonumber \\
  & \geq \left(\frac{1}{S^{*}}\right)^{K-2}
    I(W_{K-1};\allanswers^{[K-2]}|W_{[0:K-2]},F)\nonumber \\
  & \hspace{12pt}
    +\left(\frac{1}{S^{*}}+\left(\frac{1}{S^{*}}\right)^2+\cdots
    +\left(\frac{1}{S^{*}}\right)^{K-2}\right)
    L\log{|\mathcal{X}|}\label{eqn:dummy} \\
  & \geq \left(\frac{1}{S^{*}}+\left(\frac{1}{S^{*}}\right)^2+\cdots
    +\left(\frac{1}{S^{*}}\right)^{K-1}\right)
    L\log{|\mathcal{X}|}\label{eqn:capproof1},
\end{align}
					which implies that $R\geq C_{\mathcal{P}}$, defined in \eqref{eqn:capacity-arbitrary}.
					
					For any decomposable PIR scheme under collusion pattern $\mathcal{P}$ that achieves capacity, i.e., $R=C_{\mathcal{P}}$, all the inequalities in Lemmas~\ref{lemma:Ctian} and \ref{lemma:iterative-arbitrary-pir} must hold as equalities. According to the necessary and sufficient conditions in these lemmas, such codes must satisfy the properties {\bf{P1}}-{\bf{P4}}.

                    		\section{Proof of Theorem~\ref{thm:minimum-message size}}\label{appendix:Thm2proof}
					Consider a capacity-achieving uniformly decomposable PIR scheme under the collusion pattern $\mathcal{P} = \{\mathcal{T}_1, \mathcal{T}_2, \ldots, \mathcal{T}_M\}$, and let $k$ denote the index of the desired message. Recall Property {\bf{P3}}, which establishes that the residual terms $R_n^{(q_n)}=W_{[0:K-1]\setminus\{k\}}\cdot G_{n|k}^{(q_n)}$ for $n\in[0:N-1]$ are deterministic functions of $R_{\mathcal{T}_m}^{(q_{\mc{T}_m})}=(R_i^{(q_i)}: i\in\mc{T}_m)$, for each $\mathcal{T}_m\in\mathcal{P}$. 

We first claim that there exists a query set $\allqueriesinstantiation$ with $P_{\allqueries^{[k]}}(\allqueriesinstantiation)>0$ such that
\begin{align}\label{eqn:proof:Thm2proofquerset}
    H(W_{[0:K-1]\setminus\{k\}}\cdot G_{\mathcal{T}_m|k}^{(q_{\mathcal{T}_m})})\neq 0, 
\end{align}
for all $\mathcal{T}_m\in \mc{P}$, where $G_{\mathcal{T}_m|k}^{(q_{\mathcal{T}_m})}=(G_{i|k}^{(q_i)}:i\in\mathcal{T}_m)$. We prove this by contradiction. Assume that for every query set $\allqueriesinstantiation$ with $P_{\allqueries^{[k]}}(\allqueriesinstantiation)>0$, there exists some $\mathcal{T}_m\in \mc{P}$ such that $H(W_{[0:K-1]\setminus\{k\}}\cdot G_{\mathcal{T}_m|k}^{(q_{\mathcal{T}_{m}})})=0$. 

Since the terms $W_{[0:K-1]\setminus\{k\}}\cdot G_{n|k}^{(q_n)}$ for $n\in[0:N-1]\setminus\mathcal{T}_m$ are deterministic functions of $W_{[0:K-1]\setminus\{k\}}\cdot G_{\mathcal{T}_m|k}^{(q_{\mathcal{T}_m})}$, it follows that $H(W_{[0:K-1]\setminus\{k\}}\cdot G_{n|k}^{(q_n)})=0$ for all $n\in[0:N-1]$. Consequently, the responses from all servers must take the form
\begin{align}
    \allfiles\cdot G_n^{(q_n)}=W_k\cdot G_{n|[0:K-1]\setminus\{k\}}\oplus \tilde{\Delta}_n, 
\end{align}
for all $n\in[0:N-1]$, where $\tilde{\Delta}_n\in\mathcal{Y}$ are constants. This implies that the answers depend solely on the desired message $W_k$ and are independent of the other messages; such a structure cannot simultaneously satisfy the privacy and correctness requirements.

Next, consider a query set satisfying \eqref{eqn:proof:Thm2proofquerset} alongside property {\bf{P4}}, which states that the desired message terms $M_n^{(q_n)}\triangleq W_k\cdot G_{n|[0:K-1]\setminus\{k\}}^{(q_n)}$ for $n\in[0:N-1]$ are mutually independent. We will demonstrate that
\begin{align}
    H(M_S^{(q_S)})\neq 0,\ \text{for all}\ S\in\mathcal{F}(\mc{P})\label{eqn:msgnonzero},
\end{align}
where $M_S^{(q_S)}=(M_i^{(q_i)}:i\in S)$ and $\mc{F}(\mc{P})$ is as defined in \eqref{Family_for_hitting}. To see this, assume otherwise, i.e., suppose there exists an $S \in \cal{F}$ such that $H(M_n^{q_n})=0$, for all $n\in S$, i.e.,
\begin{align}
    H(W_k\cdot G_{n|[0:K-1]\setminus\{k\}}^{(q_n)})=0, \quad \text{for } n\in S, 
\end{align}
which implies that $W_k\cdot G_{n|[0:K-1]\setminus\{k\}}^{(q_n)}$ are constants for all $n\in S$. This further implies that the answers from the servers in $S$ are given by
\begin{align}
    A_n^{(q_n)}=\allfiles\cdot G_{n}^{(q_n)}=W_{[0:K-1]\setminus\{k\}}\cdot G_{n|k}^{(q_n)}\oplus\tilde{\Delta}_n\label{eqn:characproof1}, 
\end{align}
for $n\in S$, where $\tilde{\Delta}_n$ are constants. Let $\mathcal{T}_m\subseteq S$ denote the witness colluding set that certifies $S\in\mc{F}(\mc{P})$ by satisfying the condition in \eqref{Family_for_hitting}, i.e., 
\begin{align}
\forall i \in \mc{T}_m, \exists\, \mc{T}_r \in \mathcal{P}\ \text{s.t. } 
\mc{T}_r \subseteq S \setminus \{i\}.
\end{align}
Now, by \eqref{eqn:proof:Thm2proofquerset}, $R_{\mc{T}_m}^{(q_{\mc{T}_m})}$ contains a non-constant random variable, say, $H(R_i^{q_i})\neq 0$, for some $i\in\mc{T}_m\subseteq S$. By \eqref{eqn:characproof1}, this implies that $H(A_i^{(q_i)})\neq 0$. Since $\exists\, \mc{T}_r \in \mathcal{P}$ {s.t. } 
$\mc{T}_r \subseteq S \setminus \{i\}$, property \textbf{P3} implies that $A_i^{(q_i)}=R_i^{(q_i)}+\tilde{\Delta}_i$ is a deterministic function of $A_{\mc{T}_r}^{(q_{\mc{T}_r})}=R_{\mc{T}_r}^{(q_{\mc{T}_r})}+\tilde{\Delta}_{\mc{T}_r}$. This contradicts property \textbf{P1}, which requires that $A_n^{(q_n)}$ are mutually independent, since a non-constant random variable $A_i^{(q_i)}$ cannot be a deterministic function of $A_{\mc{T}_r}^{(q_{\mc{T}_r})}$ while remaining independent of it. This proves \eqref{eqn:msgnonzero}, i.e., $H(M_S^{(q_S)})\neq 0$, for all $S\in\mc{F}(\mc{P})$. Let the set of server indices with non-zero entropies be denoted by
\begin{align}
    B=\{n\in[0:N-1]: H(M_i^{(q_i)})\neq 0\}.
\end{align}
This implies that 
\begin{align}
    B\cap S\neq \emptyset, \ \text{for all}\ S\in\mc{F}(\mc{P}).
\end{align}
Let $H\subseteq [0:N-1]$ be a set of minimum cardinality that intersects every set in the family $\mc{F}(\mc{P})$. The corresponding hitting number is given by $\alpha(\mc{F}(\mc{P}))=|H|$. Since $B$ intersects every set in the family $\mc{F}(\mc{P})$, it follows that the number of non-zero entropies among $H(M_n^{q_n})$, $n\in[0:N-1]$ is at least $\alpha(\mc{F}(\mc{P}))$, i.e., 
\begin{align}\label{eqn:non-zerentries}
|B|\geq \alpha(\mc{F}(\mc{P})).
\end{align}

Therefore, we conclude that at least $\alpha(\mc{F}(\mc{P}))$ of the entropies of the random variables $M_n^{q_n}=W_k\cdot G_{n|[0:K-1]\setminus\{k\}}^{(q_n)}$, $n\in[0:N-1]$ are non-zero. Now, since the function $\varphi^{(q)}(n,i,k)$ induces a uniform probability distribution on $\mathcal{Y}$ (unless it takes a deterministic value) due to the uniform decomposability property, and by the independence property {\bf{P4}}, we have    
\begin{align}
    L\log{|\mathcal{X}|}&\geq H(W_k)\\
    &\geq H(W_k\cdot G^{(\allqueriesinstantiation)}_{[0:N-1]|[0:K-1]\setminus\{k\}})\nonumber\\
     &=\sum_{n=0}^{N-1}H(W_k\cdot G^{(q_n)}_{n|[0:K-1]\setminus\{k\}})\\
        &=\sum_{n=0}^{N-1}H(\{\varphi^{(q_n)}_{(n,i,k)}(W_k)):\forall i\in[0:\ell_n(q_n)-1]\})\nonumber \\
    &\geq \alpha(\mc{F}(\mc{P}))\log{|\mathcal{Y}|},\nonumber
\end{align}
This proves \eqref{eqn:messbound}. Noting that $|{\cal Y}|\geq 2$ for any meaningful protocol, $\eqref{eqn:messbound1}$ follows immediately. This completes the proof.

 \section{Proof of Corollary~\ref{corollory_1}}\label{appndx:cor1proof}
					\begin{proof}
							For $T$-collusion, i.e., $\mc{P}=\{\mathcal{T}\subseteq [0:N-1]: |S|=T\}$, we have
							\begin{align}\label{eqn:cor1proof1}
								\mathcal{F}(\mc{P}) = \{ S \subseteq [0:N-1] : |S| \geq T+1 \}.
							\end{align}
							We can see this as follows. If $|S|\geq T+1$, then after removing any $i\in\mathcal{T}\subseteq S$, the set $S\setminus \{i\}$ still contains at least $T$ elements left to form another colluding set. If $|S|=T$, then after removing any $i\in S$, the remaining set has size $T-1$, and thus the condition for $S\in\mc{F}(\mc{P})$ fails.
							
							We first show that $\alpha(\mc{F}(\mc{P}))\geq N-T$. Suppose, for the sake of contradiction, that there exists a hitting set $\mathcal{H}\subseteq [0:N-1]$ with $|\mathcal{H}|<N-T$. Then $S=[0:N-1]\setminus \mathcal{H}$ has cardinality at least $T+1$. Since $|S|\geq T+1$, it follows that $S\in\mc{F}(\mc{P})$. However, by construction, $S\cap \mathcal{H}=\emptyset$, which contradicts the assumption that $\mc{H}$ is a hitting set. Hence, every hitting set must have cardinality at least $N-T$, and therefore $\alpha(\mc{F}(\mc{P}))\geq N-T$. 
							
							On the other hand, consider any subset $\mathcal{H} \subseteq [0:N-1]$ of size $N-T$. For any $S\in\mc{F}(\mc{P})$, we have $|S|\geq T+1$. Therefore, 
                            \begin{align}
                                S\cap \mathcal{H}\neq \emptyset,\ \text{for all}\ S\in\mc{F}(\mc{P}),
                            \end{align}
                            which shows that $\mathcal{H}$ is a valid hitting set of size $N-T$. This implies that $\alpha(\mc{F}(\mc{P}))\leq N-T$.
                            
                            Combining the two bounds, we conclude that $\alpha(\mc{F}(\mc{P}))= N-T$. 
						\end{proof}

  \section{Proof of Corollary~\ref{corollary2}}\label{appdnx:cor2proof}
                    \begin{proof}
							Let $\{P_1,P_2,\dots,P_M\}$ be a partition of the set $[0:N-1]$ such that $|P_i|=N_i$, $i\in[1:M]$.  The family $\mc{F}(\mc{P})$ is the union of the following two types of sets.
							
							\begin{enumerate}
								\item \textbf{Type 1 Sets:} Sets that contain at least $T_i+1$ elements from a single set $P_i$:
								      \begin{align}
								      	\mathcal{F}_1 = \bigcup_{i=1}^{M} \left\{ S \subseteq [0:N-1] : |S\cap P_i| \geq T_i + 1 \right\}.
								      \end{align}
								      
								\item \textbf{Type 2 Sets:} Sets that contain at least $T_i$ elements from $P_i$ and at least $T_j$ elements from $P_j$, for some $i\neq j$: 
								      \begin{align}
								      	\mathcal{F}_2
								      	  & = \bigcup_{j \neq k} 
								      	\left\{
								      	S\subseteq [0:N-1]:
								      	\substack{
								      	 |S\cap P_j|\geq T_j, \ |S\cap P_k|\geq T_k
								      	}
								      	\right\}
								      	\label{disjoint_fmaily}
								      \end{align}
							\end{enumerate}
							The family $\mc{F}(\mc{P})=\mc{F}_1\cup \mc{F}_2$.

							First, we establish the lower bound $\alpha(\mc{F}(\mc{P})) \ge \sum_{j=1}^{J} (N_j - T_j) + (J - 1)$.
							To hit every Type~1 set associated with $P_j$, any hitting set $\mathcal{H}$ must contain at least at least $N_j - T_j$ elements from $P_j$; otherwise there would exist a subset of $P_j$ of size $T_j+1$ disjoint from $\mc{H}$, violating the hitting condition. Moreover, $\mc{H}$ can select exactly this minimum number $N_j-T_j$ of elements from at most one group. Indeed, if $\mc{H}$ were to select exactly $N_i-T_i$ elements from $P_i$, $i\in\{j,k\}$, $j\neq k$, then the corresponding sets of unselected elements $U_j\subseteq P_j$ and $U_k\subseteq P_k$ would satisfy $|U_j|=T_j$ and $|U_k|=T_k$. Consequently, the union $U_j\cup U_k$ would form a Type~2 set that is disjoint from $\mc{H}$, contradicting the hitting set requirement. Therefore, $\alpha(\mc{F}(\mc{P})) \ge \sum_{j=1}^{M} (N_j - T_j) + (M - 1)$.
							
							On the other hand, consider a specific subset $\mathcal{H}$ constructed as follows: select exactly $N_1 - T_1$ elements from  $P_1$, and select exactly $N_j - T_j + 1$ elements from every other $P_j$, $j > 1$. By construction,
                            \begin{align}
                                |\mathcal{H}|=\sum_{j=1}^M (N_j-T_j)+(M-1).
                            \end{align}
							We now verify that $\mathcal{H}$ is a valid hitting set. $\mathcal{H}$ intersects every Type~1 set, since $\mathcal{H}$ contains at least $N_j-T_j$ elements from each $P_j$, $j\in[1:M]$, and any subset of $P_j$ of cardinality $T_j+1$ must intersect $\mathcal{H}$. 

                            Consider a Type-2 set $S$ that has at least $T_j$ elements from $P_j$ and at least $T_k$ elements from $P_k$, $j\neq k$. If $j\neq 1$ and $k\neq 1$, clearly $\mc{H}$ intersects $S$ as $\mc{H}$ has exactly $N_i - T_i + 1$ elements from $P_i$, for $i\in\{j,k\}$. If $j=1$, then $\mc{H}$ contains exactly $N_1-T_1$ elements from $P_1$. However, since $\mc{H}$ contains exactly $N_k-T_k+1$ elements from $P_k$, any $T_k$-size subset of $P_k$ intersects $\mc{H}$. Thus, $\mc{H}$ intersects every Type-2 set involving $P_1$.  This implies that $\alpha(\mc{F}(\mc{P})) \leq \sum_{j=1}^{J} (N_j - T_j) + (J - 1)$.

							Combining both the bounds, we conclude that $|\mathcal{H}| = \sum_{j=1}^{J} (N_j - T_j) + (J - 1)$.

						\end{proof}

                           \section{Proof of Theorem~\ref{thm:T-contiguous}}\label{appndx:proofofthm3}
	For cyclically $T$-contiguous collusion., i.e., $\mc{P}=\{\{0,1,\dots,T-1\}, \{1,\dots,T\}, \dots, \{N-1,0,\dots,T-2\}\}$, we first show that 
    \begin{align}
    \alpha(\mc{F}(\mc{P}))=\left\lceil\frac{N}{T}\right\rceil-1,
    \end{align}
    and then present a capacity-achieving PIR scheme with message size $\frac{N}{T}-1$ when $N$ divides $T$. 

    If $N<2T$, then $\mc{F}(\mc{P})=\{[0:N-1]\}$, and $\alpha(\mc{F}(\mc{P}))=1=\left\lceil\frac{N}{T}\right\rceil-1$ when $N>T$. So, we assume that $N\geq 2T$ for the rest of the proof.

    The family $\mc{F}(\mc{P})$ is given by
    \begin{align}\label{eqn:family}
        \mc{F}(\mc{P})&=\{S\subseteq [0:N-1]:\exists \mc{T}_1,\mc{T}_2\in\mc{P}, \ \mc{T}_1\cup \mc{T}_2\subseteq S, \nonumber\\
        &\hspace{12pt}\mc{T}_1\cap \mc{T}_2=\emptyset\}.
    \end{align}
To see this, clearly, any $S\subseteq \mc{F}(\mc{P})$ satisfies the requirement in \eqref{Family_for_hitting}. Moreover, for an arbitrary $S$ satisfying the requirement in \eqref{Family_for_hitting}, let $\mathcal{T}\in\mc{P}$ be the witness colluding set such that
\begin{align}
    \forall i \in \mc{T}, \exists\, \mc{T}_r \in \mathcal{P}\ \text{s.t. } 
\mc{T}_r \subseteq S \setminus \{i\}.
\end{align}
Let $i^*\in\mathcal{T}$ be the largest index in $\mc{T}$. The set $S\setminus \{i^*\}$ should contain another colluding set which is disjoint from $\mc{T}$ because $N>2T$. 
Thus, $\mc{F}(\mc{P})$ is characterized as \eqref{eqn:family}.

\underline{$\alpha(\mc{F}(\mc{P}))\geq \left\lceil\frac{N}{T}\right\rceil-1$:}  We partition the set $[0:N-1]$ into $r=\left\lceil\frac{N}{T}\right\rceil$ pairwise disjoint contiguous subsets $B_1,B_2,\dots,B_r$, where $|B_i|=T$, $i\in[1:r-1]$, $|B_r|\leq T$ (as $\frac{N}{T}$ may not be an integer). Suppose, for the sake of contradiction, there exists a hitting set $\mc{H}$ with $|\mc{H}|< \left\lceil\frac{N}{T}\right\rceil-1=r-1$. Then $\mc{H}$ can intersect at most $r-2$ of the sets $B_1,\dots,B_r$. 

If $T$ divides $N$, we have $|B_r|=T$. Then it follows that there exists two distinct indices $i\neq j$ such that $H\cap B_i=\emptyset$ and $\mc{H}\cap B_j=\emptyset$. Consequently, $H\cap (B_i\cup B_j)=\emptyset$. However, since $B_i\cup B_j\in\mc{F}(\mc{P})$, this contradicts the assumption that $\mc{H}$ is a hitting set.

If $T$ does not divide $N$, we have $|B_r|<T$. If $\mc{H}\cap B_r\neq \emptyset$, then again there exists $i\neq j$ such that $\mc{H}\cap (B_i\cup B_j)=\emptyset$, which is a contradiction. So, we assume $\mc{H}\cap B_r=\emptyset$. Then $\mc{H}$ can intersect at most $r-2$ of the sets $B_1,\dots,B_{r-1}$. Let $\mc{H}\cap B_{r-1}=\emptyset$. Then $\mc{H}$ must have at least one element from each of $B_i$, $i\in[1:r-2]$ to maintain the hitting set property. Let us analyze the placement of these elements. Say, the element fro m$B_1$ is in the $i^{\text{th}}$ position in that set.  Let $|B_r|=N-(r-1)T=y=T-x$. If $i\geq x+1$, then the first $x$ elements of $B_1$ together with $B_r$ form a contiguous set of size $T$, i.e., a colluding set $\mc{T}$. Now, $\mc{H}\cap (\mc{T}\cup B_r)=\emptyset$, again violating the hitting set requirement. If instead $i<x+1$, then each element selected from $B_2,\dots,B_{r-2}$ must occur at or before the $i^{\text{th}}$ position in their respective sets; otherwise a similar contradiction arises. This constrains the overall position of selected elements. But now, the elements to the right of the chosen element in $B_{r-2}$, all the elements in $B_{r-1}$ and $B_r$, and the elements to the left of the chosen element in $B_1$ contains a contiguous set of size $2T$ because $T-t+T+y+t=2T+y$. This contains two disjoint colluding sets $\mc{T}_i$ and $\mc{T}_j$, leading to a contradiction to the hitting set requirement as $\mc{T}_i\cup \mc{T}_j\in\mc{F}(\mc{P})$. The other cases, i.e., when $\mc{H}\cap B_j=\emptyset$, $j\in[1:r-2]$, can be dealt in a similar manner. In particular, when $\mc{H}\cap B_1=\emptyset$, we can construct a contiguous set of size $2T$ not overlapping with $\mc{H}$ leading to a contradiction to the hitting set requirement. When $\mc{H}\cap B_i=\emptyset$, $i\in[2:r-2]$, we can construct a $T$-contiguous set not overlapping with $\mc{H}$, which together with $B_i$ contradicts the hitting set requirement. Therefore, $\alpha(\mc{F}(\mc{P}))\geq \left\lceil\frac{N}{T}\right\rceil-1$.  

\underline{$\alpha(\mc{F}(\mc{P}))\leq \left\lceil\frac{N}{T}\right\rceil-1$:}  Consider the partition of $[0:N-1]$ into $r=\left\lceil\frac{N}{T}\right\rceil$ pairwise disjoint contiguous subsets $B_1,B_2,\dots,B_r$, where $|B_i|=T$, $i\in[1:r-1]$, $|B_r|\leq T$ (since $\frac{N}{T}$ may not be an integer). Now, construct a set $\mc{H}$ by selecting the first element from each of the sets $B_1,\dots,B_{r-1}$, and none from $B_r$. Note that these elements are equally spaced at intervals of length $T$. We now verify that this is a hitting set. Consider an arbitrary set $S\in\mc{F}(\mc{P})$, so $\mc{T}_j\cup \mc{T}_j\subseteq S$,  where $\mc{T}_i\cap \mc{T}_j=\emptyset$, for $i\neq j$. We claim that at least one of $\mc{T}_i$ or $\mc{T}_j$ must overlap with the union $B_1\cup \dots \cup B_{r-2}\cup \{n\}$, where $n$ is the first element of $B_{r-1}$ Suppose not, both $\mc{T}_i$ and $\mc{T}_j$ are entirely contained in $(B_{r-1}\setminus \{n\})\cup B_r$, then this leads to a contradiction because $|(B_{r-1}\setminus \{n\})\cup B_r|<2T-1$. Now, since $\mc{T}_i$ or $\mc{T}_j$ overlaps with of $B_1\cup \dots \cup B_{r-2}\cup \{n\}$, clearly, $\mc{H}$ intersects $\mc{T}_j\cup \mc{T}_j$ because the elements of $\mc{H}$ are equally spaced at intervals of $T$ across the entire index set. This shows that $\alpha(\mc{F}(\mc{P}))\leq\left\lceil\frac{N}{T}\right\rceil-1$. 

Combining both the bounds, we have $\alpha(\mc{F}(\mc{P}))=\left\lceil\frac{N}{T}\right\rceil-1$.

\underline{An Achievable Scheme with Message Size $\frac{N}{T}-1$.} We present a capacity achieving PIR scheme under the cyclically $T$-contiguous collusion pattern with message size $L=\frac{N}{T}-1$, assuming $T$ divides $N$. The central idea is to identify a subset of $\frac{N}{T}$ servers such that no two servers within this subset participate in a common colluding set. To construct this, partition the server index set $[0:N-1]$ into $r=\frac{N}{T}$ pairwise disjoint contiguous subsets $B_1,B_2,\dots,B_r$, where $|B_i|=T$, $i\in[1:r]$. From each $B_i$, select its first server (i.e., with the smallest index), and construct a new set $\mathcal{N}=\{0, T,\dots, (r-1)T\}$. This gives us $|\mathcal{N}| = \frac{N}{T}$ servers that are equally spaced at intervals of $T$ across the entire index set. Due to the cyclic $T$-contiguous structure of the collusion pattern, no two servers in $\mathcal{N}$ are part of any colluding set, i.e., they are mutually non-colluding. We relabel these servers as $n' = 0, 1, \dots, \frac{N}{T} - 1$. On this reduced set of servers, we apply the achievable scheme of \cite[Section~III]{TianSC19}, which is valid as no two servers collude among $n' = 0, 1, \dots, \frac{N}{T} - 1$. The resulting scheme achieves the rate
\begin{align}\label{eqn:capctycont}
C=\left(
1 + \frac{1}{(\frac{N}{T})} + \left(\frac{1}{(\frac{N}{T})}\right)^2 + \cdots +
\left(\frac{1}{(\frac{N}{T})}\right)^{K-1}
\right)^{-1}
\end{align}
with message size $L=\frac{N}{T}-1$.
This rate is known to match exactly the capacity of PIR under cyclically $T$-contiguous collusion~\cite{SunJ18_MDS,YaoLK21}. For completeness, and more importantly, because we will require it in the proof of Theorem~\ref{them:disjointcontiguous}, we now show that that the general capacity expression in \eqref{eqn:capacity-arbitrary} reduces to the expression in \eqref{eqn:capctycont} when the collusion pattern is $T$-cyclically contiguous.

 The incidence matrix is given by a single
circulant matrix $\mathbf{B}_{\cal{P}}$, where each column
of $\mathbf{B}_{\cal{P}}$ contains exactly $T$ ones and $N-T$ zeros.

We consider the following linear program:
\begin{align}\label{eqn: LPapndx}
\max_{\mathbf{y}} \quad & \sum_{i=1}^N y_i \\
\text{s.t.} \quad & \mathbf{B}^{\mathsf T}_{\cal{P}} \mathbf{y} \leq \mathbf{1},\nonumber \\
& \mathbf{y} \geq \mathbf{0}\nonumber.
\end{align}

Due to the circulant structure of $\mathbf{B}_{\cal{P}}$, the constraints can be
written explicitly as
\begin{align}
& y_1 + y_2 + \cdots + y_T \le 1, \nonumber\\
& y_2 + y_3 + \cdots + y_{T+1} \le 1, \nonumber\\
& \hspace{2cm}\vdots \nonumber\\
& y_N + y_1 + \cdots + y_{T-1} \le 1, \label{eq:contiguous_constraints}\\
& y_i \ge 0,\quad i=1,\dots,N. \nonumber
\end{align}

Since each variable $y_i$ appears in exactly $T$ constraints in
\eqref{eq:contiguous_constraints}, the symmetric solution
$y_i = 1/T$ for all $i$ satisfies all constraints with equality and is optimal.
Consequently $S^{*}$ will be $\sum\limits_{i=0}^{N} y_i = \frac{N}{T}$. Therefore the capacity is given by the expression in \eqref{eqn:capctycont}.

                         \section{Proof of Theorem~\ref{them:disjointcontiguous}}\label{appndx:thm4proof}
For $(N_1,\dots,N_M ; T_1,\dots,T_M)$-cyclically contiguous collusion, we first show that $\alpha(\mc{F}(\mc{P}))=\sum_{i=1}^M\left\lceil\frac{N_i}{T_i}\right\rceil-1$, and then present a capacity-achieving PIR scheme with message size $\sum_{i=1}^M\frac{N_i}{T_i}-1$, when $N_i$ divides $T_i$, for $i\in[1:M]$.

 Consider a partition of the full server index set $[0:N-1]$ into $M$ disjoint groups $P_1, P_2, \dots, P_M$, where $|P_i| = N_i$, and the collusion pattern within each $P_i$ is cyclically contiguous of size $T_i$. Let $\mc{P}$ denote the set of all the colluding sets. We have
 \begin{align}\label{eqn:familydisjoint}
        \mc{F}(\mc{P})&=\{S\subseteq [0:N-1]:\exists \mc{T}_1,\mc{T}_2\in\mc{P}, \ \mc{T}_1\cup \mc{T}_2\subseteq S, \nonumber\\
        &\hspace{12pt}\mc{T}_1\cap \mc{T}_2=\emptyset\}.
    \end{align}
    Notice that $\mc{T}_1$ and $\mc{T}_2$ in \eqref{eqn:familydisjoint} can be from same group $P_j$ or can be from different groups $P_j$ and $P_k$, $j\neq k$.
			
								First, we show that $\alpha(\mc{F}(\mc{P}))\geq\sum_{i=1}^M \lceil \frac{N_i}{T_i} \rceil - 1$. For contradiction, suppose that there exists a hitting set $\mc{H}$ with $|\mc{H}|<\sum_{i=1}^M \lceil \frac{N_i}{T_i} \rceil - 1=\sum_{i=1}^M(\lceil\frac{N}{T}\rceil-1)+M-1$. By Theorem~\ref{thm:T-contiguous},  $\mc{H}$ must contain at least $\lceil\frac{N_i}{T_i}\rceil-1$ elements from each group $P_i$, so there exists at least two colluding sets $\mc{T}_1$ and $\mc{T}_2$ in different groups $P_j$ and $P_k$ such that $\mc{H}\cap (\mc{T}_1\cup \mc{T}_2)=\emptyset$, which violates the definition of a hitting set. Therefore, $\alpha(\mc{F}(\mc{P}))\geq\sum_{i=1}^M \lceil \frac{N_i}{T_i} \rceil - 1$.
								
								On the other hand, construct a set $\mc{H}$ by taking unions of the following sets.
                                \begin{itemize}
                                    \item For each $i\in[1:M]$, include a hitting set $\mc{H}_i$ for the family of sets containing unions of disjoint colluding sets entirely within in $P_i$. This can be done by selecting the first element from each of the $\lceil \frac{N_i}{T_i} \rceil - 1$ full $T_i$-sized blocks in a contiguous partition of $P_i$, omitting the last (possibly incomplete) block, following the construction used in the upper bound proof of Theorem~\ref{thm:T-contiguous}.
                                    \item Additionally, for each $i\in[2:M]$, select one element from the last block in the partition of $P_i$, specifically, the first element in that block. No such selection is made for $P_1$.
                                \end{itemize}
                                The total size of $\mc{H}$ is therefore $|\mathcal{H}|$ = $\sum_{i=1}^M (\lceil \frac{N_i}{T_i} \rceil - 1) + (M-1)$ =$\sum_{i=1}^M (\lceil \frac{N_i}{T_i} \rceil )- 1$. The set $\mc{H}$ intersects every set of the form $\mc{T}_1\cup \mc{T}_2$ with $\mc{T}_1\cap \mc{T}_2=\emptyset$. If $\mc{T}_1, \mc{T}_2$ both lie within the same group $P_i$, the set $\mc{H}_i$ already guarantees intersection due to the construction in the first bullet above. If $\mc{T}_1$ and $\mc{T}_2$ lie in different groups $P_j$ and $P_k$ (with $j \neq k$), then one of them must intersect its corresponding $\mc{H}_j$ (or $\mc{H}_k$), or the additional selected elements from the last blocks ensure intersection, due to the second bullet above.
         This implies that $\alpha(\mc{F}(\mc{P}))\leq\sum_{i=1}^M \lceil \frac{N_i}{T_i} \rceil - 1$.
                                
								Combining both the bounds, we conclude that $\alpha(\mc{F}(\mc{P}))=\sum_{i=1}^M \lceil \frac{N_i}{T_i} \rceil - 1$.

\underline{An Achievable Scheme with Optimal Message Size} 
We present a capacity-achieving PIR scheme under $(N_1,\dots,N_M ; T_1,\dots,T_M)$-cyclically contiguous collusion, assuming that each $T_i$ divides $N_i$. We identify a subset of servers from each disjoint partition such that no two selected servers belong to a common colluding set, either within the same group or across different groups. 

Specifically, consider a partition of the full server index set $[0:N-1]$ into $M$ disjoint groups $P_1, P_2, \dots, P_M$, where $|P_i| = N_i$, and the collusion pattern within each $P_i$ is cyclically contiguous of size $T_i$. For each partition $P_i$, divide it into $r_i = \frac{N_i}{T_i}$ disjoint contiguous blocks $B_{i,1}, B_{i,2}, \dots, B_{i,r_i}$, where $|B_{i,j}| = T_i$. From each $B_{i,j}$, select the first server (i.e., the one with the smallest index), and define the set
\[
\mathcal{N} = \bigcup_{i=1}^M \{ \text{first servers in } B_{i,1}, B_{i,2}, \dots, B_{i,r_i} \}.
\]
This gives a total of $\sum_{i=1}^M \frac{N_i}{T_i}$ servers in $\mathcal{N}$, one from each block, such that no two servers in $\mathcal{N}$ collude. This holds because:
\begin{itemize}
    \item within each $P_i$, the spacing of $T_i$ ensures non-collusion under the cyclic $T_i$-contiguous pattern;
    \item across different partitions, collusion is restricted within each $P_i$ by the disjoint structure.
\end{itemize}

We relabel the servers in $\mathcal{N}$ as $n' = 0, 1, \dots, \left( \sum_{i=1}^M \frac{N_i}{T_i} \right) - 1$, and apply the achievable scheme of \cite[Section~III]{TianSC19}, which is valid since no two servers collude. The resulting scheme achieves a rate

\begin{align}\label{eqn:capacitydisjointcontig}
C
=
\left(
1 + \frac{1}{S^\ast} + \left(\frac{1}{S^\ast}\right)^2 + \cdots +
\left(\frac{1}{S^\ast}\right)^{K-1}
\right)^{-1},
\end{align}
where 
\begin{align}
S^\ast
= \sum_{j=1}^{J} \frac{N_j}{T_j},
\end{align}
with a message size of $L=\sum_{i=1}^M\left\lceil\frac{N_i}{T_i}\right\rceil-1$. Using \eqref{eqn:capacity-arbitrary}, we prove below that the expression in exactly matches the capacity of PIR under $(N_1,\dots,N_M ; T_1,\dots,T_M)$-cyclically contiguous collusion.

                    The incidence matrix corresponding to the collusion pattern is
\[
\mathbf{B}_{\cal{P}} =
\begin{pmatrix}
B_1 & 0 & \cdots & 0 \\
0 & B_2 & \cdots & 0 \\
\vdots & \vdots & \ddots & \vdots \\
0 & 0 & \cdots & B_J
\end{pmatrix},
\]
where each block $B_j$ is an $N_j \times N_j$ circulant matrix, and each column of $B_j$ contains exactly $T_j$ ones and $N_j - T_j$ zeros.

The linear program corresponding to this incidence matrix decomposes into $M$ independent linear programs, each associated with a $T_j$-contiguous collusion pattern. Therefore, it suffices to solve the linear program for a single $T_j$-contiguous collusion pattern. This subproblem corresponds exactly to the linear program described in \eqref{eqn: LPapndx} and \eqref{eq:contiguous_constraints}.

Using the solution of the linear program from the proof of Theorem~\ref{thm:T-contiguous}, it is straightforward to see that the optimal solution to the for the current linear program is
\[
\mathbf{y}^\ast =
\big(
\underbrace{\tfrac{1}{T_1},\ldots,\tfrac{1}{T_1}}_{N_1},
\underbrace{\tfrac{1}{T_2},\ldots,\tfrac{1}{T_2}}_{N_2},
\ldots,
\underbrace{\tfrac{1}{T_J},\ldots,\tfrac{1}{T_J}}_{N_J}
\big)^{\mathsf T}.
\]

The corresponding optimal value is
\[S^{*}=
\sum_{i=1}^{N} y_i^\ast
= \sum_{j=1}^{M} \frac{N_j}{T_j}.
\]

Therefore, the capacity is given by the expression in \eqref{eqn:capacitydisjointcontig} with $S^*=\sum_{j=1}^{M} \frac{N_j}{T_j}$.

                    \section{Missing Details in Remark~\ref{remark2}}\label{recovering_rank_equations}

Noting that \cite[Equations (119) and (117)]{miki2025necessary} correspond exactly to the equalities in \eqref{eq:3prime} and \eqref{eq:equality_for_shearer}, specialized to the $T$-collusion pattern for linear PIR schemes, the properties stated in Theorem~\ref{thm:characterizing-properties} recover \eqref{eqn:tcolrec1} and \eqref{eqn:tcolrec2} by proceeding with the same remaining steps as in \cite[Proof of Theorem~5]{miki2025necessary}.

\section{Missing Details in Remarks~\ref{remark5} and \ref{remark7}}\label{apndx:messagesize-contiguous}
In \cite[Appendix~A]{YaoLK21}, the authors provide a capacity-achieving scheme with message size $L$ chosen such that $(\frac{1}{S})^K(S-1)^{k-1}Ly_n,\text{is an integer }\forall k\in[1:K],n\in[1:N]$, under arbitrary collusion pattern $\mc{P}$, where $\mathbf{y}^*=(y_1,\dots,y_n)$ is an optimal solution to the linear program in \eqref{eq:LP1}.
\subsection{For Cyclically $T$-Contiguous Collusion}
For cyclically $T$-contoiguous collusion, the solution to the linear program in \eqref{eq:LP1} is given by $y_{n}= \frac{1}{T}$, $n\in[0:N-1]$, with the optimal value $S^*= \frac{N}{T}$.
To ensure that $(\frac{1}{S})^K(S-1)^{k-1}Ly_n$ is an integer for all $k$ and $n$, it suffices to choose $L=N^{K}$ in the general case.

\subsection{For $(N_1,\dots,N_M ; T_1,\dots,T_M)$-cyclically contiguous collusion}
The solution to the linear program in \eqref{eq:LP1} for this case is 
\[
\mathbf{y}^\ast =
\big(
\underbrace{\tfrac{1}{T_1},\ldots,\tfrac{1}{T_1}}_{N_1},
\underbrace{\tfrac{1}{T_2},\ldots,\tfrac{1}{T_2}}_{N_2},
\ldots,
\underbrace{\tfrac{1}{T_J},\ldots,\tfrac{1}{T_J}}_{N_J}
\big)^{\mathsf T},
\]
with the optimal value $S = \sum \limits_{j=1}^{M} \frac{N_{j}}{T_{j}}$.
To ensure $(\frac{1}{S})^K(S-1)^{k-1}Ly_n$ is integer for all $k$ and $n$, it suffices to choose 
\begin{equation}
L
=
\left(
\sum_{i=1}^{M}
\frac{ \left( \prod_{j=1}^{M} T_j \right) N_i }{ T_i }
\right)^{K} \nonumber
\end{equation}

in the general case.

\fi
					
					\bibliographystyle{IEEEtran}
					\bibliography{bibliofile}

\end{document}